\documentclass[journal]{IEEEtran}
\ifCLASSINFOpdf
\else
\fi

\usepackage{siunitx}
\usepackage{graphicx}
\usepackage{float}
\usepackage{amsfonts}
\usepackage{amsmath}
\usepackage{subcaption}
\usepackage{amsthm}    
\usepackage{amssymb}
\usepackage{xcolor}
\usepackage{mathtools}
\usepackage{algorithm}
\usepackage{algorithmic}
\usepackage{enumitem}
\usepackage{cite}
\newtheorem{theorem}{Theorem}

\theoremstyle{definition}
\newtheorem{assumption}{Assumption}

\newtheorem{lemma}{Lemma}

\theoremstyle{remark}
\newtheorem{remark}{Remark}

\theoremstyle{definition}

\begin{document}
%
\title{Adaptive Observer for a Class of Systems with Switched Unknown Parameters Using DREM}
%
%
%

\author{Tong~Liu, Zengjie~Zhang,~\IEEEmembership{Member, IEEE},~ Fangzhou~Liu$^*$,~\IEEEmembership{Member, IEEE},~
        and~Martin~Buss,~\IEEEmembership{Fellow,~IEEE}
\thanks{
* Corresponding author.

Tong Liu, Fangzhou Liu, Martin Buss are with the Chair of Automatic Control Engineering,
        Technische Universit{\"a}t M{\"u}nchen, Munich, 80333, Germany
        fangzhou.liu@tum.de}
\thanks{Zengjie Zhang is with with the School of Engineering, the University of British Columbia, 1137 Alumni Avenue, Kelowna, V1V 1V7, BC, Canada}}

%
%

\markboth{Journal of \LaTeX\ Class Files,~Vol.~14, No.~8, August~2015}%
{Shell \MakeLowercase{\textit{et al.}}: Bare Demo of IEEEtran.cls for IEEE Journals}
%



\maketitle

\begin{abstract}
In this note, we develop an adaptive observer for a class of nonlinear systems with switched unknown parameters to estimate the states and parameters simultaneously.
The main challenge lies in how to eliminate the disturbance effect of zero-input responses caused by the switching on the parameter estimation. These responses depend on the unknown states at switching instants (SASI) and constitute an additive disturbance to the parameter estimation, which obstructs parameter convergence to zero. Our solution is to treat the zero-input responses as excitations instead of disturbances. This is realized by first augmenting the system parameter with the SASI and then developing an estimator for the augmented parameter using the \textit{dynamic regression extension and mixing} (DREM) technique. Thanks to its property of element-wise parameter adaptation, the system parameter estimation is decoupled from the SASI.
As a result, the estimation errors of system states and parameters converge to zero asymptotically. 
Furthermore, the robustness of the proposed adaptive observer is guaranteed in the presence of disturbances and noise. A numerical example validates the effectiveness of the proposed approach.
\end{abstract}


%
\IEEEpeerreviewmaketitle

\section{Introduction}

Over the last decades, a lot of efforts have been devoted to the simultaneous state and parameter estimation of dynamical systems with adaptive observers.
Early results of adaptive observer design focus on systems, which can be transformed into canonical form\cite{ioannou1996robust,kreisselmeier1977adaptive,marino1992global,besanccon2000remarks}.
In \cite{zhang2002adaptive}, a new adaptive observer is proposed for a class of time-varying systems, which does not require the system to be transformed into canonical form. 
Extensions of this adaptive observer to nonlinearly parameterized systems\cite{farza2009adaptive}, systems nonlinear in the parameters\cite{tyukin2013adaptive}, and stochastic systems\cite{zhang2018adaptive} are reported. These approaches also enjoy a wide range of applications, e.g., biodiesel fueled engines\cite{zhao2014adaptive}, lithium-ion batteries\cite{wang2014adaptive}, and antilock braking systems\cite{aguado2018switched}. 
The fundamental idea of most referenced adaptive observers is to construct a linear regression equation (LRE) by utilizing filtering operation on known signals depending on the inputs and outputs. The LRE enables the application of various adaptive parameter estimation approaches to estimate the unknown parameters. Then the state observation is conducted with the estimated parameter. It is worth pointing out that the LRE captures the forced response of the filtered system whereas the zero-input response stemming from the filtering operation is treated as a disturbance and mostly disregarded, as it is exponentially decaying and does not destroy the convergence of the estimation errors. 

Despite the above-mentioned advances, the existing adaptive observers focus on systems with constant unknown parameters and cannot be applied to systems with switched unknown parameters. In practice, the operation conditions of most plants may change, which cannot be modeled using constant parameters. Instead, systems with switched unknown parameters are used to depict this phenomenon\cite{liberzon2003switching,di2009hybrid,kersting2019recursive,moustakis2018fault}. This motivates us to explore adaptive observers for systems with switched unknown parameters. 
The main challenge in this regard lies in the disturbance effect of zero-input responses caused by the switching. These responses are products of unknown states at switching instants (SASI) and the known transient terms. They act as a non-vanishing unknown additive disturbance term to the LRE and prevent the parameter estimation error from converging to zero. 
This problem cannot be solved by using conventional adaptive observers. However, the recently proposed adaptive observers using dynamic regressor extension and mixing (DREM)\cite{ortega2015parameter,pyrkin2019adaptive,ortega2021generalized} provide new inspiration to cope with the zero-input responses.
These approaches treat the initial state as an unknown parameter and transform the state observation problem into the parameter estimation problem of the initial state. The key feature of these DREM-based approaches is the ability to ensure the \textit{element-wise} parameter adaptation, namely, the adaptation of each element of the estimated parameter is decoupled from each other. Although methods in \cite{ortega2015parameter,pyrkin2019adaptive,ortega2021generalized} are not eligible for switched systems, we are inspired by their element-wise adaptation property and develop an adaptive observer for systems with switched unknown parameters by using DREM, which overcomes the challenge of the treatment of zero-input responses.

The main contribution of this technical note is that we develop an adaptive observer for a class of systems with switched unknown parameters. We achieve asymptotic convergence of state and parameter estimation errors despite the presence of the non-vanishing zero-input responses. Furthermore, the robustness of the proposed adaptive observer is analyzed for the case with noise and disturbances. We emphasize the novelty of the technical route through which we cope with the zero-input responses. Specifically, distinct from the most adaptive observers for systems with constant unknown parameters\cite{kreisselmeier1977adaptive,marino1992global,besanccon2000remarks,zhang2002adaptive,farza2009adaptive,tyukin2013adaptive,zhang2018adaptive,zhao2014adaptive,wang2014adaptive,aguado2018switched}, where the zero-input response is viewed as a disturbance and disregarded, we exploit the known information of the zero-input responses to construct the regressor and augment the system parameters with SASI (the unknown part of zero-input responses).
To decouple the system parameter estimation and the evolution of SASI, we propose a DREM-based parameter estimator for the augmented parameter. Thanks to its element-wise adaptation property, the asymptotic convergence of state and parameter estimation errors is achieved.

The rest of this note is structured as follows. The problem formulation and preliminaries are given in Sec. \ref{sec: pre}. The proposed adaptive observer is depicted in Sec. \ref{sec: main}. The numerical validation is shown in Sec.\ref{sec: sim}. In Sec. \ref{sec: con}, the conclusion is given and future work is discussed.
\section{Problem Formulation and Preliminaries}
\label{sec: pre}
Consider the uncertain nonlinear switched system
\begin{subequations}\label{eqn: plant_ss}
\begin{align}
        \dot{x}&=A x + B u + \Psi(y,u) \theta_{\sigma(t)}^*, \label{eqn: plant_ss1}\\
        y&=C x, \label{eqn: plant_ss2}
\end{align}
\end{subequations}
where $x \in \mathbb{R}^n$ is the state vector, $y \in \mathbb{R}$ denotes the output signal, and $u \in \mathbb{R}$ represents the input signal of the system. $A \in \mathbb{R}^{n \times n}, B \in \mathbb{R}^{n \times 1}, C \in \mathbb{R}^{1 \times n}$ are known constant matrices. The nonlinearity $\Psi(y,u) \in \mathbb{R}^{n \times m}$ is a known time-varying matrix depending on the output $y$ and the input $u$. The switched system (\ref{eqn: plant_ss}) has $s\in\mathbb{N}^+$ subsystems and $\theta_{\sigma(t)}^*\in \mathbb{R}^{m}$ denotes the switched unknown parameter vector with $\theta_{\sigma(t)}^* \in \{\theta_{1}^*,\theta_{2}^*,\cdots,\theta_{s}^*\}$. The switching signal $\sigma(t): [0,\infty) \to \mathcal{I} \triangleq \{1,2,\cdots,s\}$ is a known piecewise constant function. It governs, which subsystem is activated. Namely, for $\sigma(t)=i, i\in \mathcal{I}$, we have $\theta_{\sigma(t)}^*=\theta_i^*$ and we say that $i$-th subsystem is activated at time $t$. To characterize the switching instants, let the set of switching time instants represented by $\{t_1, t_2, \cdots, t_k, \cdots\}$ for $k\in \mathbb{N}^+$ and the initial time instant denoted by $t_0$. 

Different from the well-known system form with constant unknown parameters studied in\cite{zhang2002adaptive,aguado2018switched,zhang2018adaptive}, system (\ref{eqn: plant_ss}) depicts systems with switched unknown parameters. It is of practical interest as it can model many engineering applications operating in multiple modes such as aircraft wing\cite{zavieh2013intersection}, mechanical systems with friction\cite{van2008tracking} and backlash\cite{yamada2018piecewise}.

The problem to be solved in this paper is formulated as follows:
given a switched system (\ref{eqn: plant_ss}) with 
the unknown subsystem parameters $\theta_i^*, i\in\mathcal{I}$, design an adaptive observer based on the input $u$ and the output $y$ to simultaneously estimate the system state $x$ and the subsystem parameters $\theta_i^*, i\in \mathcal{I}$ with asymptotic convergence of the estimation errors.

For the observer design in this paper, we make the following assumptions:

\begin{assumption}
\label{asp: bounded}
The state $x(t)$, the output $y(t)$, the input $u(t)$ and the parameters $\theta_i^*$ are bounded.
i.e., $x(t) \in X, y(t) \in Y, u(t) \in U, \forall t\geq 0$ and $\theta_i^* \in \Theta_i, i\in\mathcal{I}$ with $X \in \mathbb{R}^n, Y \in \mathbb{R}, U\in\mathbb{R},\Theta_i\in\mathbb{R}^{m}$ being compact sets.
\end{assumption}
Assumption \ref{asp: bounded} is a common assumption in adaptive observer design problem\cite{farza2009adaptive}\cite{liu2020switched}. 



\section{Main Results}
\label{sec: main}
In this section, we introduce the proposed adaptive observer to solve the above-mentioned problem. We first derive the LRE of (\ref{eqn: plant_ss}) and redefine the role of the transient zero-input responses caused by switching by augmenting the unknown parameters with the SASI.
Then, we develop a DREM-based parameter estimator to decouple the parameter estimation from the SASI. 
Based on the estimated parameter, we design the state observer and conduct the robustness analysis.
\subsection{Derivation of LRE}
Let us start by transforming the system (\ref{eqn: plant_ss}) into a LRE. The goal of this step is to establish an algebraic relation between the unknown parameters and the known signals. 
We rewrite (\ref{eqn: plant_ss1}) as 
\begin{equation}
    \label{eqn: plant_ss_KC}
        \dot{x}=(A-KC) x + B u + Ky + \Psi(y,u) \theta_{\sigma(t)}^*
\end{equation}
with $K\in\mathbb{R}^{n \times 1}$ being an output feedback gain such that $(A-KC)$ is Hurwitz. The time response of $x(t)$ for the interval $t\in[t_k,t_{k+1}), k \in \mathbb{N}$, in which $\theta_{\sigma(t)}^*$ remains constant, can be written as
\begin{align}
    \begin{split}
    \label{eqn: x_solution}
        x(t)=\Phi(t,t_k)x(t_k)+&\int_{t_k}^t\Phi(t,\tau)(Bu(\tau)+Ky(\tau))\textrm{d}\tau\\ &+\int_{t_k}^t\Phi(t,\tau)\Psi\theta_{\sigma(t)}^*\textrm{d}\tau,
    \end{split}
\end{align}
where $\Phi(t,\tau)$ is the state transition matrix associated with $(A-KC)$.
From (\ref{eqn: x_solution}) we can see that three components constitute the solution of $x$ for $t \in [t_k, t_{k+1}), k\in \mathbb{N}$: the zero-input response associated with the SASI $x(t_k)$, the forced response driven by the known signal $B u + Ky$ and the forced response driven by the switched uncertain part $\Psi \theta_{\sigma}^*$. 
These two forced responses can also be described by two auxiliary signals $x_u$ and $x_{\theta}$ generated by the following dynamics
\begin{subequations}\label{eqn: decompose_x}
\begin{align}
    \dot{x}_u&=(A-KC) x_u + B u + Ky, \hspace{0.5em} x_u(t_k)=0, \\
    \dot{x}_{\theta}&=(A-KC) x_{\theta} + \Psi\theta_{\sigma(t)}^*, \quad \quad x_{\theta}(t_k)=0, k\in \mathbb{N}
\end{align}
\end{subequations}
Therefore, for $t\in [t_k, t_{k+1}), k\in \mathbb{N}$, equation (\ref{eqn: x_solution}) becomes
\begin{equation}
\label{eqn: decompose_x_2}
    x=\Phi(t,t_k)x(t_k)+x_u+x_{\theta}.
\end{equation}
Furthermore, let the signal matrix $\Upsilon(t)\in\mathbb{R}^{n \times m}$ be generated by the following dynamics
\begin{equation}
\label{eqn: decompose_x_3}
    \dot{\Upsilon}=(A-KC) \Upsilon + \Psi, \quad \Upsilon(t_k)=0, k\in \mathbb{N},
\end{equation}
from which one obtains $\Upsilon(t)=\int_{t_k}^t \Phi(t,\tau)\Psi\textrm{d}\tau, t\in [t_k, t_{k+1})$. This together with the solution of $x_{\theta}$ (the third term in (\ref{eqn: x_solution})) leads to
\begin{equation}
\label{eqn: x_theta}
    x_{\theta}=\Upsilon\theta_{\sigma(t)}^*.
\end{equation}

The signals $x_u, x_\theta, \Upsilon$ can be respectively viewed as filtered signals of $Bu+Ky, \Psi \theta_{\sigma(t)}^*, \Psi$ with the filter parameter $(A-KC)$. $x_{\theta}$ is unknown and to be estimated. $x_u, \Upsilon$ are known and will be used later for the adaptive observer design.
\begin{remark}
In the adaptive observer design for non-switching systems\cite{zhang2002adaptive}, it suffices to specify zero initial states of the auxiliary filtered signal $x_u, \Upsilon$ at $t_0$. In contrast, these signals in our context are reset to zero at each time instant $t_k$ ($k=0$ for the initial instant and $k \in \mathbb{N}^+$ for switching instants, see (\ref{eqn: decompose_x}), (\ref{eqn: decompose_x_3})). Such reset is essential for a clear decomposition of $x$ into the zero-input response (see $\Phi(t,t_k)x(t_k)$ in (\ref{eqn: decompose_x_2})) and forced responses (see $x_u, x_{\theta}$ in (\ref{eqn: decompose_x_2})) for every continuous interval $[t_k, t_{k+1}), k\in\mathbb{N}$. 
\end{remark}
Recall that the goal of this section is to establish an algebraic relation between the unknown parameters and known signals by transforming the system (\ref{eqn: plant_ss}) into a LRE. To achieve this, we take (\ref{eqn: x_theta}) into (\ref{eqn: decompose_x_2}), move the known signal $x_u$ to the left side of (\ref{eqn: decompose_x_2}), and multiply both sides with $C$, which yields for $t \in [t_k, t_{k+1}), k\in \mathbb{N}$ the following LRE
\begin{equation}
\label{eqn: parameterization}
    z=C\Upsilon\theta_{\sigma(t)}^*+C\Phi(t,t_k)x(t_k).
\end{equation}
with $z = y-Cx_u$. 
\begin{remark}
For systems without switching, the LRE would become $z=C\Upsilon\theta^*+C\Phi(t,t_0)x(t_0), t \in [t_0, \infty)$ with $\theta^*$ being the constant unknown parameter. In most of the approaches of the current line of research \cite{kreisselmeier1977adaptive,marino1992global,besanccon2000remarks,zhang2002adaptive,farza2009adaptive,tyukin2013adaptive,zhang2018adaptive,zhao2014adaptive,wang2014adaptive,aguado2018switched}, only $z=C\Upsilon\theta^*$ is considered and the zero-input response $C\Phi(t,t_0)x(t_0)$ is disregard due to its exponentially decaying property. The work \cite{aranovskiy2015flux} provides rigorous asymptotic convergence analysis when $C\Phi(t,t_0)x(t_0)$ is not neglected. 
As oppose to the decaying property in these papers, the zero-input responses $C\Phi(t,t_k)x(t_k)$ in our case build a non-vanishing disturbance signal under intermittent switching, as each switch triggers a zero-input response depending on the SASI $x(t_k), k\in\mathbb{N}^+$. 
Consequently, the effect of these transients on the parameter estimation cannot be neglected. 
\end{remark}
The disturbance effect of $C\Phi(t,t_k)x(t_k), k\in \mathbb{N}$ is the main obstacle to obtain the asymptotic convergence of the estimation errors. Observe that $C\Phi(t,t_k)x(t_k)$ consists of the known part $C\Phi(t,t_k)$ and the unknown part $x(t_k)$. Our solution concept is to treat the SASI $x(t_k)$ as a part of the unknown parameters and view $C\Phi(t,t_k)$ as a part of the regressor such that we can make full use of this known signal for the parameter estimation. Namely, we rewrite (\ref{eqn: parameterization}) as
\begin{equation}
\label{eqn: parameterization_aug}
    z=\nu^T \bar{\theta}^*(t)
\end{equation}
with the augmented parameter and regressor
\begin{align}
\begin{split}
\label{eqn: bar_theta}
    \nu^T=[C\Upsilon, C\Phi(t,t_k)] &\in \mathbb{R}^{1 \times (m+n)},\\
    \bar{\theta}^*(t)=
    \begin{bmatrix}
        \theta_{\sigma(t)}^*\\
        x^*(t)
    \end{bmatrix}&\in \mathbb{R}^{(m+n)},
\end{split}
\end{align}
where $x^*(t) \in \mathcal{X} \triangleq \{x(t_1), x(t_2), \cdots\, x(t_k), \cdots \}$ and $x^*(t)=x(t_k)$ for $t \in [t_k, t_{k+1}), k\in \mathbb{N}$. Viewing over the whole time interval $t\in [t_0,\infty)$, $x^*(t)$ is a piecewise constant vector. 

The augmented parameter vector $\bar{\theta}^*(t)$ in (\ref{eqn: bar_theta}) is a switched parameter, which remains constant within each interval $[t_k, t_{k+1})$ for $k \in \mathbb{N}$ and switches at each $t_k$. It consists of two parts: the parameters to be estimated $\theta_i^*$, $i\in \mathcal{I}$ and the state at each switching instant $x^*(t) \in \mathcal{X}$. The cardinality of $\mathcal{I}$ is $s$ while the cardinality of $\mathcal{X}$ is unknown. 
Due to the mismatch of the cardinalities, it is necessary to develop a parameter estimator, which enables separable adaptations of $\theta_{\sigma(t)}^*$ and $x^*$. To realize this, we propose a DREM-based parameter estimator as it can achieve element-wise adaptation of the parameter.

\subsection{DREM-based Parameter Estimator}
The LRE (\ref{eqn: parameterization_aug}) is derived by incorporating the filterd signals $x_u, \Upsilon$ with the filter parameter $(A-KC)$. Based on this derivation,
the first step of DREM is to create $m+n$ LRE by using a set of filters with distinct filter parameters $(A-K_j C), j\in\{1,2,\cdots,m+n\}$ instead of using a single filter with $(A-K C)$. $K_j$ are designed such that $(A-K_j C)$ are Hurwitz. So we repeat the derivation from (\ref{eqn: plant_ss_KC}) to (\ref{eqn: bar_theta}) and replace $K$ with $K_j, j\in\{1,2,\cdots,m+n\}$. This leads to
\begin{equation}
    z_j=\nu_j^T \bar{\theta}^*(t)
\end{equation}
with
\begin{align}
    \begin{split}
    \label{eqn: def_nu_j}
        z_j &= y-C x_{uj},\\
        \nu_j^T &=[C\Upsilon_j, C\Phi_j(t,t_k)],
    \end{split}
\end{align}
where
\begin{align*}
    \begin{split}
        \dot{x}_{uj}&=(A-K_j C) x_{uj} + B u + K_jy, \hspace{0.5em} x_{uj}(t_k)=0,\\
        \dot{\Upsilon}_j&=(A-K_j C) \Upsilon_j + \Psi, \hspace{0.5em} \quad \quad \quad \quad \Upsilon_j(t_k)=0, k\in \mathbb{N}
    \end{split}
\end{align*}
and $\Phi_j(t,\tau)$ denotes the state transition matrix associated with $(A-K_j C)$. We rewrite the $m+n$ LRE in matrix form and obtain
\begin{equation}
\label{eqn: extended_regressor}
    Z_f=N^T \bar{\theta}^*(t)
\end{equation}
with
\begin{equation}
\label{eqn: def_N}
    Z_f=
    \begin{bmatrix}
        z_1\\
        \vdots\\
        z_{m+n}
    \end{bmatrix},
    \quad
    N^T=
    \begin{bmatrix}
        \nu_1^T\\
        \vdots\\
        \nu_{m+n}^T
    \end{bmatrix}.
\end{equation}
The second step of DREM suggests multiplying the extended regression equation (\ref{eqn: extended_regressor}) with the adjoint of the extended regressor matrix $N^T$, denoted by $\textrm{adj}(N^T)$. This leads to
\begin{equation}
\label{eqn: extended_regressor2}
    \textrm{adj}(N^T) Z_f= \textrm{adj}(N^T) N^T \bar{\theta}^*(t) = \textrm{det}(N) \bar{\theta}^*(t)
\end{equation}
with $\textrm{det}(\cdot)$ denoting the determinant of a matrix. Let $\Delta = \textrm{det}(N) \in \mathbb{R}$ and $\bar{\mathcal{Z}} = \textrm{adj}(N^T) Z_f \in \mathbb{R}^{m+n}$. From (\ref{eqn: extended_regressor2}) we obtain 
\begin{equation}
\label{eqn: extended_regressor3}
    \bar{\mathcal{Z}}=\Delta \bar{\theta}^*(t),
\end{equation}
As $\Delta$ is a scalar, (\ref{eqn: extended_regressor3}) leads to $m+n$ separate scalar regression equations
\begin{equation}
\label{eqn: extended_regressor4}
    \bar{\mathcal{Z}}=
    \begin{bmatrix}
        \bar{\mathcal{Z}}_1\\
        \vdots\\
        \bar{\mathcal{Z}}_m\\
        \bar{\mathcal{Z}}_{m+1}\\
        \vdots\\
        \bar{\mathcal{Z}}_{m+n}\\
    \end{bmatrix}
    =
    \Delta 
    \begin{bmatrix}
        \theta_{1\sigma(t)}^*\\
        \vdots\\
        \theta_{m\sigma(t)}^*\\
        x_{1}^*(t)\\
        \vdots\\
        x_{n}^*(t)
    \end{bmatrix}
    =\Delta \bar{\theta}^*(t),
\end{equation}
where $\bar{\mathcal{Z}}_j$ is the $j$-th element of the vector $\bar{\mathcal{Z}}$, $\theta_{j i}^*$ represents the $j$-th element of $\theta_i^*$, and $x_j^*$ denotes the $j$-th element of $x^*$. 

We define the following indicator functions
\begin{equation}
    \chi_i(t)=
    \begin{cases}
        1, \quad \textrm{if } \sigma(t)=i,\\
        0, \quad \textrm{otherwise.}
    \end{cases}
\end{equation}
Let $\hat{{\theta}}_{ji}\in\mathbb{R}^{m}$ be the estimated value of ${\theta}^*_{ji}$ and let $\hat{\theta}_i=[\hat{{\theta}}_{1i},\cdots,\hat{{\theta}}_{mi}]^T$.
The adaptation of the estimated parameter follows the adaptation law
\begin{equation}
\label{eqn: adaptive law}
    \dot{\hat{\theta}}_{ji}=
        \gamma_i \chi_i \Delta (\bar{\mathcal{Z}}_j-\Delta \hat{\theta}_{ji})
\end{equation}
for $i\in\mathcal{I}, j\in\{1,\cdots,m\}$. $\gamma_i \in \mathbb{R}^{+}$ is a positive scaling factor. This adaptation law gives the parameter error equation
\begin{equation}
\label{eqn: parameter_error}
    \dot{\tilde{\theta}}_{ji}=
        -\gamma_i \chi_i \Delta^2 \tilde{\theta}_{ji}
\end{equation}
for $\tilde{\theta}_{ji}$ being the $j$-th element of $\tilde{\theta}_i$ with $\tilde{\theta}_i = \hat{\theta}_i-\theta_i^*$. The indicator function $\chi_i$ in the adaptation law (\ref{eqn: adaptive law}) indicates that the value of $\hat{\theta}_i$ remains constant when subsystem $i$ is inactive and $\hat{\theta}_i$ is adapted during the active phase of subsystem $i$.

As underscored in the introduction, the conceptual highlight of this paper is to convert the role of the zero-input responses from disturbances to excitations. To better understand this concept, we observe from (\ref{eqn: def_nu_j}) and (\ref{eqn: def_N}) that the known part of the zero-input response $C\Phi_j(t,t_k)$ constitutes a part of the regressor matrix $N$, whose determinant $\Delta$ further drives the adaptation of the parameter estimation errors (see (\ref{eqn: parameter_error})). Furthermore, the element-wise adaptation property of DREM ensures that the evolution of the unknown part of the zero-input response $x(t_k)$ does not affect the adaptation of the estimated system parameters $\hat{\theta}_i$.
\begin{remark}
Simulation results in \cite{zhang2002adaptive,zhang2018adaptive} show that adaptive observers proposed for non-switched systems have the tolerance for rare switch of the parameters at the expense of transient parameter estimation errors after each switch.
Due to these transient parameter estimation errors, provable asymptotic convergence of parameter and state estimation errors cannot be established. Moreover, performance degradation may occur when the time between two successive switches of parameters is not long enough to let the transients converge. 
These problems still persist in the very recently proposed adaptive observer for switched systems\cite{liu2020switched}. One feature that distinguishes our method from these methods is that each subsystem has its own estimated parameter $\hat{\theta}_i, i\in \mathcal{I}$. The estimated parameter $\hat{\theta}_i$ is only adapted when $i$-th subsystem is activated. Otherwise, $\hat{\theta}_i$ is frozen and is retained as the initial value for the next active period for $i$-th subsystem. Therefore, asymptotic convergence of parameter estimation errors can be achieved without suffering from transient errors after each switch.
\end{remark}
\subsection{Adaptive State Observer}
After obtaining the parameter adaptation law (\ref{eqn: adaptive law}) for the estimated parameters $\hat{\theta}_i$, the adaptive state observer to estimate the state $x$ is given based on $\hat{\theta}_i$
\begin{align}
\begin{split}
\label{eqn: adaptive observer}
    \dot{\hat{x}}&=A \hat{x}+B u+ \Psi(y,u)\hat{\theta}_{\sigma(t)} + K(y-\hat{y}),\\
    \hat{y}&=C\hat{x}
\end{split}
\end{align}
where $\hat{x}, \hat{y}$ denote the estimated state and output, respectively. $K \in \mathbb{R}^{n \times 1}$ is to be chosen such that $(A-KC)$ is Hurwitz.

\begin{assumption}
\label{asp: L2}
The scalar signal $\chi_i\Delta$ satisfies $\chi_i\Delta \notin \mathcal{L}_2, \forall i \in \mathcal{I}$.
\end{assumption}
The performance of the proposed adaptive observer (\ref{eqn: adaptive law}), (\ref{eqn: adaptive observer}) can be summarized below.
\begin{theorem}
\label{thm: ideal}
Consider the switched system (\ref{eqn: plant_ss}) with unknown parameters $\theta_i^*, i\in \mathcal{I}$ and the adaptive observer (\ref{eqn: adaptive observer}) with the adaptation law (\ref{eqn: adaptive law}). If Assumption \ref{asp: bounded} and Assumption \ref{asp: L2} hold, then we have the parameter estimation error $\tilde{\theta}_i(t) \to 0, \forall i \in \mathcal{I}$ as $t \to \infty$ and the state estimation error $\tilde{x}(t)=\hat{x}(t)-x(t) \to 0$ as $t \to \infty$. 
\end{theorem}
\begin{proof}
From (\ref{eqn: parameter_error}), we have 
\begin{equation}
    \tilde{\theta}_{ji}(t)=\mathrm{e}^{-\gamma_i \int_{t_0}^t \chi_i(s)\Delta^2(s)\textrm{d}s}\tilde{\theta}_{ji}(t_0).
\end{equation}

Since $\chi_i\Delta \notin \mathcal{L}_2, i \in \mathcal{I}$, we have $\tilde{\theta}_{ji} \to 0$ and therefore, $\tilde{\theta}_i \to 0$ as $t \to \infty$. From (\ref{eqn: plant_ss}) and (\ref{eqn: adaptive observer}) we obtain for $\tilde{x}=\hat{x}-x$
\begin{equation}
    \dot{\tilde{x}}=(A-KC)\tilde{x}+\Psi \tilde{\theta}_{\sigma(t)}.
\end{equation}
Since $\tilde{\theta}_i \to 0, \forall i\in \mathcal{I}$ for $t \to 0$ and $(A-KC)$ is Hurwitz, it leads to $\tilde{x} \to 0$ as $t \to 0$.
\end{proof}
\begin{remark}
Regarding the zero-input response, the underlying concept of the adaptive observers for non-switched systems\cite{kreisselmeier1977adaptive,marino1992global,besanccon2000remarks,zhang2002adaptive,farza2009adaptive,tyukin2013adaptive,zhang2018adaptive,zhao2014adaptive,wang2014adaptive,aguado2018switched},\cite{aranovskiy2015flux} or adaptive control for switched systems\cite{sang2012adaptive,liu2021output} is treating the zero-input response as a disturbance, regardless of whether neglecting it\cite{kreisselmeier1977adaptive,marino1992global,besanccon2000remarks,zhang2002adaptive,farza2009adaptive,tyukin2013adaptive,zhang2018adaptive,zhao2014adaptive,wang2014adaptive,aguado2018switched} or including it in the stability analysis\cite{aranovskiy2015flux,sang2012adaptive,liu2021output}. Distinct from this concept, we provide a new perspective that the zero-input responses can be utilized as excitations to promote the parameter estimation such that asymptotic convergence of state and parameter estimation errors can be achieved.
\end{remark}
\begin{remark}
The DREM-based adaptive observers in \cite{ortega2015parameter,pyrkin2019adaptive} augment the system parameters with the initial state. Then the state estimation is established based on the identified initial state through an open-loop integration\cite{ortega2015parameter} or a non-fragile algebraic equation\cite{pyrkin2019adaptive}. As oppose to these approaches, the purpose of augmenting the system parameters with SASI in this note is to exploit the known information of the zero-input responses in the regressor whereas the estimation of SASI is not of interest and disregarded.
\end{remark}
\begin{remark}
For the convergence analysis of the parameter estimation in switched systems, it is a common condition that each subsystem is activated intermittently\cite{kersting2019recursive}\cite{liu2021output}\cite{yuan2016adaptive}, i.e., for any $i \in \mathcal{I}$ and $T \in \mathbb{R}^+$, there exist $\bar{t} >T$ and $t_\delta \in \mathbb{R}^+$ such that $\sigma(t)=i$ for $t \in [\bar{t}, \bar{t}+t_\delta)$. In our paper, this condition is implicitly included in the condition $\chi_i\Delta \notin \mathcal{L}_2, \forall i \in \mathcal{I}$ (Assumption \ref{asp: L2}). Specifically, if there exists subsystem $l \in \mathcal{I}$ that is not intermittently activated, then there exists $\check{t}\geq 0$ such that $\sigma(t)\neq l$ and $\chi_l(t)=0, \forall t \in [\check{t},\infty)$. This would lead to $\chi_l \Delta \in \mathcal{L}_2$, which contradicts with the condition $\chi_i\Delta \notin \mathcal{L}_2, \forall i \in \mathcal{I}$. Therefore, Assumption \ref{asp: L2} implies that every subsystem is activated intermittently.
\end{remark}

\subsection{Robustness Analysis}
In this section, we study the robustness of the adaptive observer when applying it to systems with disturbances and noise. Consider the system
\begin{subequations}\label{eqn: plant_ss_corrupt}
\begin{align}
        \dot{x}&=A x + B u + \Psi(y,u) \theta_{\sigma(t)}^*+\omega, \label{eqn: plant_ss_corrupt_1}\\
        y&=C x,\label{eqn: plant_ss_corrupt_2}\\
        \bar{y}&=y + v, \label{eqn: plant_ss_corrupt_3}
\end{align}
\end{subequations}
where $\omega \in \mathbb{R}^n$ represents the state disturbance. $\bar{y}$ denotes the measured output with $v \in \mathbb{R}$ being the measurement noise. $\omega$ and $v$ are bounded, i.e., $\|\omega\| \leq \omega_0, |v| \leq v_0$ for some constants $\omega_0, v_0 \in \mathbb{R}^+$ with $\|\cdot\|$ denoting the norm and $|\cdot|$ representing the absolute value of a scalar.

With regards to the nonlinear function $\Psi(y,u)$ in (\ref{eqn: plant_ss_corrupt}), we make the additional assumption as follows:
\begin{assumption}
\label{asp: Lipschitz}
The function $\Psi(y,u)$ in (\ref{eqn: plant_ss_corrupt}) is Lipschitz with respect to $y$. That is, there exists positive constant $L_{\Psi} \in \mathbb{R}^+$ such that $\forall u \in U$ and $y, \bar{y} \in Y$ we have $\|\Psi(y,u)-\Psi(\bar{y},u)\| \leq L_{\Psi} |y - \bar{y}| = L_{\Psi} |v| \leq L_{\Psi} v_0$.
\end{assumption}

To study how $\omega$ and $v$ affect the stability, we rederive the error equation (\ref{eqn: parameter_error}). We start by rewriting (\ref{eqn: plant_ss_corrupt_1}) as 
\begin{equation}
    \label{eqn: plant_ss_KC_corrupt}
        \dot{x}=(A-K_j C) x + B u + K_jy + \Psi(y,u) \theta_{\sigma(t)}^* + \omega
\end{equation}
for some $K_j \in\mathbb{R}^{n \times 1}, j\in\{1,2,\cdots,m+n\}$ such that $(A-K_j C)$ is Hurwitz. The time response of $x(t)$ for $t\in[t_k,t_{k+1}), k \in \mathbb{N}$ can be written as
\begin{align}
    \begin{split}
    \label{eqn: x_solution_corrupt}
        x(t)=&\Phi_j(t,t_k)x(t_k)+\int_{t_k}^t\Phi_j(t,\tau)(Bu+K_j y)\textrm{d}\tau\\ &+\int_{t_k}^t\Phi_j(t,\tau)\Psi(y,u)\theta_{\sigma}^*\textrm{d}\tau+\int_{t_k}^t\Phi_j(t,\tau)\omega \textrm{d}\tau.
    \end{split}
\end{align}
We use the measured output $\bar{y}$ to generate the filtered signals
\begin{align*}
    \begin{split}
        \dot{x}_{uj}&=(A-K_j C) x_{uj} + B u + K_j \bar{y}, \hspace{0.5em} x_{uj}(t_k)=0,\\
        \dot{\Upsilon}_j&=(A-K_j C) \Upsilon_j + \Psi(\bar{y},u), \hspace{0.5em} \quad \quad \Upsilon_j(t_k)=0, k\in \mathbb{N}
    \end{split}
\end{align*}
Let $z_j=\bar{y}-C x_{uj}$, which gives for $t \in [t_k, t_k+1), k \in \mathbb{N}$
\begin{align}
    z_j=Cx+v-C(\int_{t_k}^t\Phi_j(t,\tau)(Bu+K_j y+K_j v)\textrm{d}\tau).
\end{align}
Substituting $x$ with (\ref{eqn: x_solution_corrupt}) yields
\begin{equation}
\label{eqn: z_corrupt}
    z_j=C\Phi_j(t,t_k)x(t_k)+C\Upsilon_j \theta_{\sigma(t)}^* + d_j
\end{equation}
with the disturbance-related term $d_j$ being expressed by
\begin{align}
\begin{split}
\label{eqn: d_j}
    d_j=&C\int_{t_k}^t\Phi_j(t,\tau)(\Psi(y,u)-\Psi(\bar{y},u))\theta_{\sigma(t)}^*\textrm{d}\tau\\
    &+C\int_{t_k}^t\Phi_j(t,\tau)\omega\textrm{d}\tau + v - C \int_{t_k}^t\Phi_j(t,\tau)K_j v \textrm{d}\tau.
\end{split}
\end{align}
With the regressor $\nu_j^T =[C\Upsilon_j, C\Phi_j(t,t_k)]$, equation (\ref{eqn: z_corrupt}) can be written in the linear regression form
\begin{equation}
    z_j=\nu_j^T \bar{\theta}^*(t) + d_j.
\end{equation}
Stacking $m+n$ equations yields
\begin{equation}
\label{eqn: extended_regressor_corrupt}
    Z_f=N^T \bar{\theta}^*(t) + d
\end{equation}
with $d=[d_1,d_2,\cdots,d_{m+n}]^T$.
Multiplying both sides with $\textrm{adj}(N^T)$ leads to 
\begin{equation}
    \bar{\mathcal{Z}}=\Delta \bar{\theta}^*(t)+\bar{D}
\end{equation}
with $\bar{D}=\textrm{adj}(N^T) d$. We apply the same parameter adaptation law as in the ideal case (\ref{eqn: adaptive law}) and obtain the element-wise parameter error equation
\begin{equation}
\label{eqn: parameter_error_corrupt}
    \dot{\tilde{\theta}}_{ji}=
        -\gamma_i \chi_i \Delta^2 \tilde{\theta}_{ji} + \gamma_i \chi_i \Delta \bar{D}_j
\end{equation}
where $i \in \mathcal{I}, j\in \{1,2,\cdots,m\}$, $\bar{D}_j$ is the $j$-th element of $\bar{D}$. Let $D$ be a vector of the first $m$ elements of $\bar{D}$, i.e., $D = [\bar{D}_1, \bar{D}_2, \cdots, \bar{D}_m]^T$. We can write (\ref{eqn: parameter_error_corrupt}) into the vector form
\begin{equation}
\label{eqn: parameter_error_vector_corrupt}
    \dot{\tilde{\theta}}_{i}=
        -\gamma_i \chi_i \Delta^2 \tilde{\theta}_{i} + \gamma_i \chi_i \Delta D
\end{equation}
The adaptive observer is constructed based on the measured output $\bar{y}$ and (\ref{eqn: adaptive observer}) now becomes
\begin{align}
\begin{split}
\label{eqn: adaptive observer_corrupt}
    \dot{\hat{x}}&=A \hat{x}+B u+ \Psi(\bar{y},u)\hat{\theta}_{\sigma(t)} + K(\bar{y}-\hat{y}),\\
    \hat{y}&=C\hat{x}
\end{split}
\end{align}
\begin{assumption}[Persistence of Excitation (PE) \cite{ioannou1996robust}]
\label{asp: PE}
The scalar signal $\chi_i\Delta$ is PE $\forall i \in \mathcal{I}$. That is, there exist constants $\alpha_0, T_0>0$ such that
\begin{equation*}
   \frac{1}{T_0}\int_{t}^{t+T_0} \chi_i(\tau)\Delta^2(\tau) \textrm{d}\tau \geq \alpha_0, \quad \forall t \geq t_0.
\end{equation*}
\end{assumption}
The following theorem summaries the robustness of the proposed adaptive observer in the presence of disturbances and noise.
\begin{theorem}
\label{thm: robust}
Consider the switched system (\ref{eqn: plant_ss_corrupt}) with unknown parameters $\theta_i^*, i\in \mathcal{I}$ and the adaptive observer (\ref{eqn: adaptive observer_corrupt}) with the adaptation law (\ref{eqn: adaptive law}). If Assumption \ref{asp: bounded}, Assumption \ref{asp: Lipschitz}, and Assumption \ref{asp: PE} hold, then the parameter estimation error $\tilde{\theta}_i$ and the state estimation error $\tilde{x}$ converge to the residual set
\begin{equation}
\label{eqn: robust_bound}
    \mathcal{R}_e=\{\tilde{x}, \tilde{\theta}_i \big| \|\tilde{x}\|+ \|\tilde{\theta}_i\| \leq \mu \sup_t\|\Delta D\|+c \}
\end{equation}
for some positive constants $\mu, c \in \mathbb{R}^+$.
\end{theorem}
The proof of Theorem \ref{thm: robust} can be seen in Appendix \ref{sec: apd_proof}.
\begin{remark}
Compared to the disturbance-free case, Theorem \ref{thm: robust} requires a stronger excitation condition that $\chi_i\Delta$ is PE, which is instrumental to ensure the boundedness of $\tilde{x}$ and $\tilde{\theta}_i$. In case the PE condition cannot be satisfied in some circumstances, robust modifications such as projections and leakages\cite{ioannou1996robust} can be applied to the adaptation law (\ref{eqn: adaptive law}). The modified adaptation law together with the boundedness of $D$ in (\ref{eqn: parameter_error_vector_corrupt}) (proved in Appendix \ref{sec: apd_proof}) would lead to the boundedness of $\tilde{x}$ and $\tilde{\theta}_i$ (see \cite[chp. 9.2]{ioannou1996robust}).
\end{remark}


\section{Numerical Examples}
\label{sec: sim}
In this section, the proposed adaptive observer is validated through a numerical example of the chaotic oscillator adjusted from the literature\cite{zhang2007output,chen2015adaptive}. Its system equation is given by
\begin{equation}
\label{eqn: Chua}
    \begin{cases}
        \dot{x}_1=p_0 (-x_1+x_2-g(x_1))\\
        \dot{x}_2=x_1-x_2+x_3\\
        \dot{x}_3=-q_0 x_2-r_0 x_3
    \end{cases}
\end{equation}
where $p_0=10, q_0=16, r_0=0.0385$ are known parameters. Let $x=[x_1, x_2, x_3]^T$ be the state vector. The system output $y=x_1$ is measurable and $x_2, x_3$ are to be estimated. The function $g(x_1)$ is a piecewise linear function
\begin{equation}
    g(x_1)=
    \begin{cases}
        -0.7143 x_1 -0.4286, &\textrm{for } x_1 \geq 1\\
        -1.1429 x_1, & \textrm{for } |x_1|<1\\
        -0.7143 x_1 +0.4286, &\textrm{for } x_1 \leq -1
    \end{cases}
\end{equation}
Therefore, the system (\ref{eqn: Chua}) can be written in form of (\ref{eqn: plant_ss}) with
\begin{equation}
    A=
    \begin{bmatrix}
        -p_0 & p_0 & 0\\
        1 & -1 & 1\\
        0 & -q_0 & -r_0
    \end{bmatrix},
    \quad
    \Psi=-p_0 
    \begin{bmatrix}
        y & 1\\
        0 & 0\\
        0 & 0
    \end{bmatrix},
\end{equation}
$B=[0,0,0]^T$ and $C=[1,0,0]$. The initial state of the system is $x(0)=[2.88,-0.066,-2.12]^T$. The nominal parameters to be estimated are $\theta_1^*=[-0.7143, -0.4286]^T$, $\theta_2^*=[-1.1429, 0]^T$, and $\theta_3^*=[-0.7143, 0.4286]^T$. The switching signal $\sigma(t)=1$ for $x_1(t) \in \Omega_1 = \{x_1 |\,  x_1 \geq 1\}$, $\sigma(t)=2$ for $x_1(t) \in \Omega_2 = \{x_1 | \, |x_1| < 1\}$, and $\sigma(t)=3$ for $x_1(t) \in \Omega_3 = \{x_1 |\,  x_1 \leq -1\}.$

Now we evaluate the estimation performance of our proposed adaptive observer in the ideal case. The filter parameters $\{K_j\}_{j=1}^5$ are chosen as $K_1=[0,-1,-15]^T$, $K_2=[-2,2.5,20]^T$, $K_3=[-2,0.1,1]^T$, $K_4=[-0.4,-0.4,-8]^T$, $K_5=[-8,6.5,18]^T$ and $K$ in the state observer (\ref{eqn: adaptive observer}) is $K=[-2,2.5,20]^T$. The initial value of the observer $\hat{x}(0)=0$. We specify the scaling factors $\gamma_i=10, i=\{1,2,3\}$. The switching signal is show in Fig. \ref{fig: mode_DREM}. We can observe the intermittent switching, namely, every mode is repeatedly activated. Fig. \ref{fig: Delta_DREM} shows the evolution of integrals $\int_0^t \chi_i(s) \Delta^2(s) \textrm{d}s, i \in \mathcal{I}$, from which one can conclude $\chi_i\Delta \notin \mathcal{L}_2$. The norm of the parameter estimation error for each subsystem $\|\tilde{\theta}_i\|$ is shown in Fig.\ref{fig: theta_DREM}, where dashed sections represent the inactive phase and solid sections represent the active phase. As it can be seen from Fig.\ref{fig: theta_DREM}, the value of $\|\tilde{\theta}_i\|$ stays unchanged during the inactive phase whereas it, thanks to the use of DREM, decreases monotonically during the active phase. Furthermore, The estimated parameters of all subsystems converge to $0$. The element-wise state estimation 
is shown in Fig.\ref{fig: x1_DREM}, Fig.\ref{fig: x2_DREM}, and Fig.\ref{fig: x3_DREM}, respectively. The red solid lines display estimated states and the blue dashed lines represent real states. One can observe that the state estimation errors also converge to $0$, this together with the parameter convergence validates the theoretical results of Theorem \ref{thm: ideal} that the proposed method is able to eliminate the disturbance effect of the zero-input responses and achieves asymptotic convergence of state and parameter estimation errors.
\begin{figure}[h]
\centering
    \subcaptionbox{\label{fig: mode_DREM}}{\includegraphics[width=0.23\textwidth]{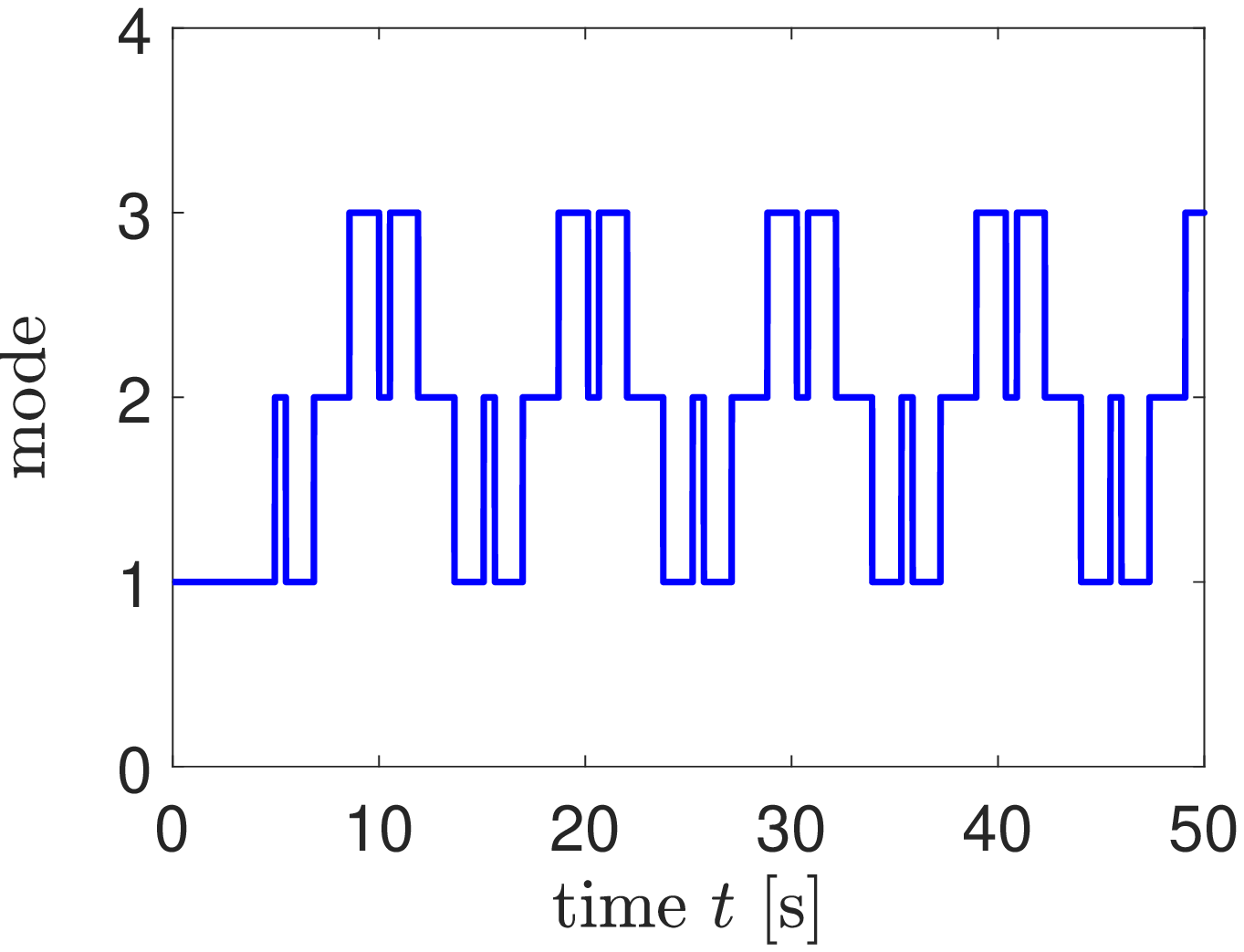}}
    \subcaptionbox{\label{fig: Delta_DREM}}{\includegraphics[width=0.23\textwidth]{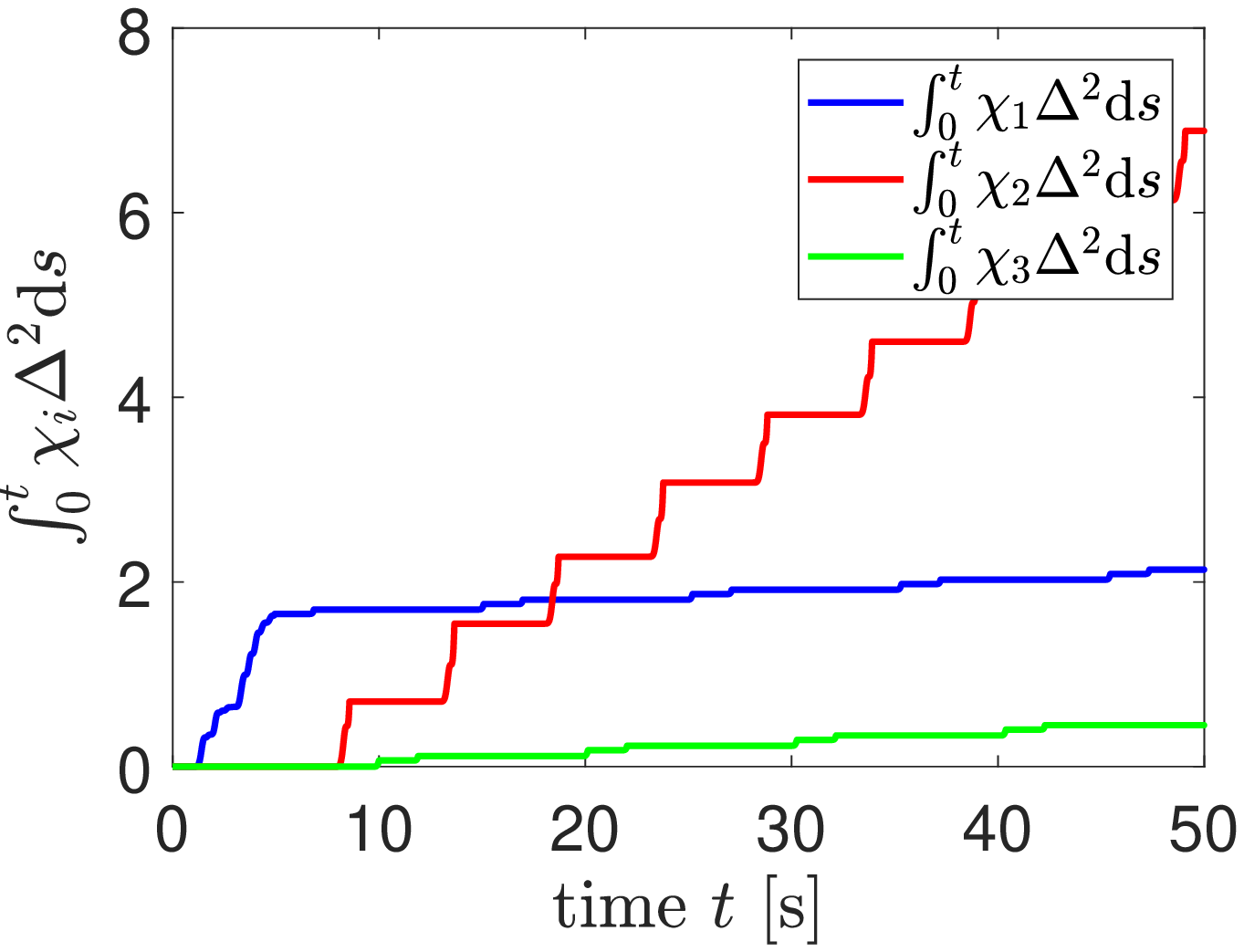}}
    \subcaptionbox{\label{fig: theta_DREM}}{\includegraphics[width=0.23\textwidth]{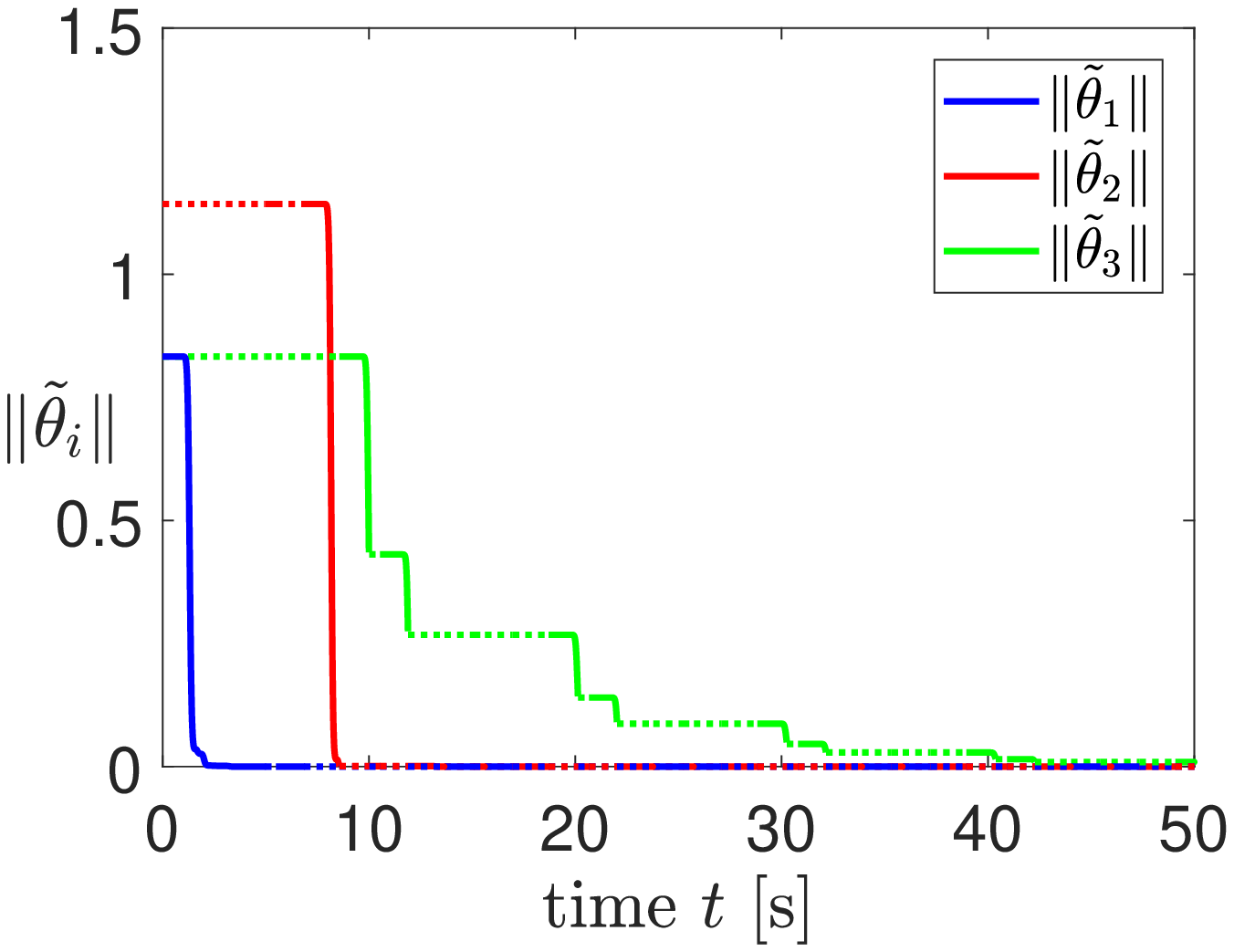}}
    \subcaptionbox{\label{fig: x1_DREM}}{\includegraphics[width=0.23\textwidth]{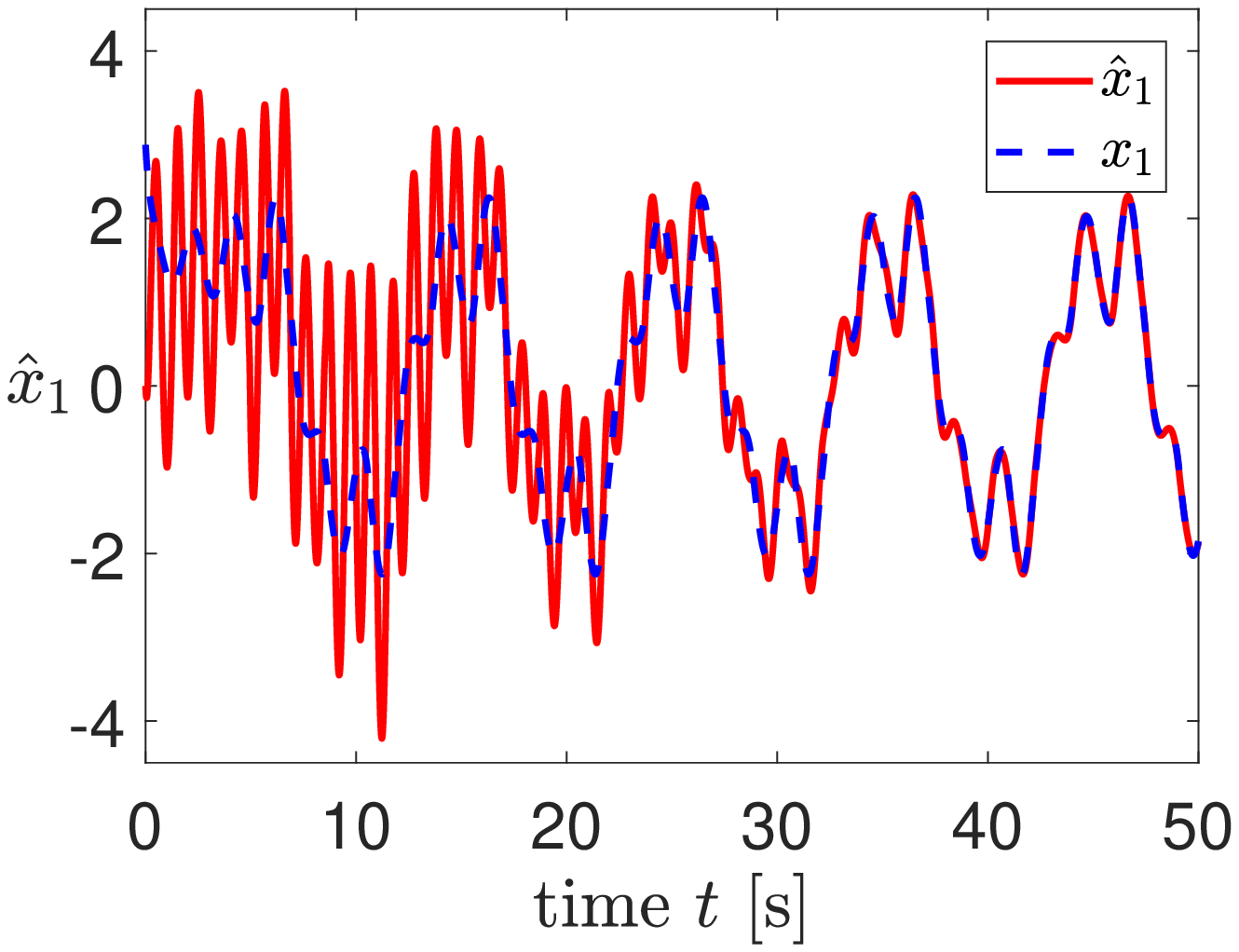}}
    \subcaptionbox{\label{fig: x2_DREM}}{\includegraphics[width=0.23\textwidth]{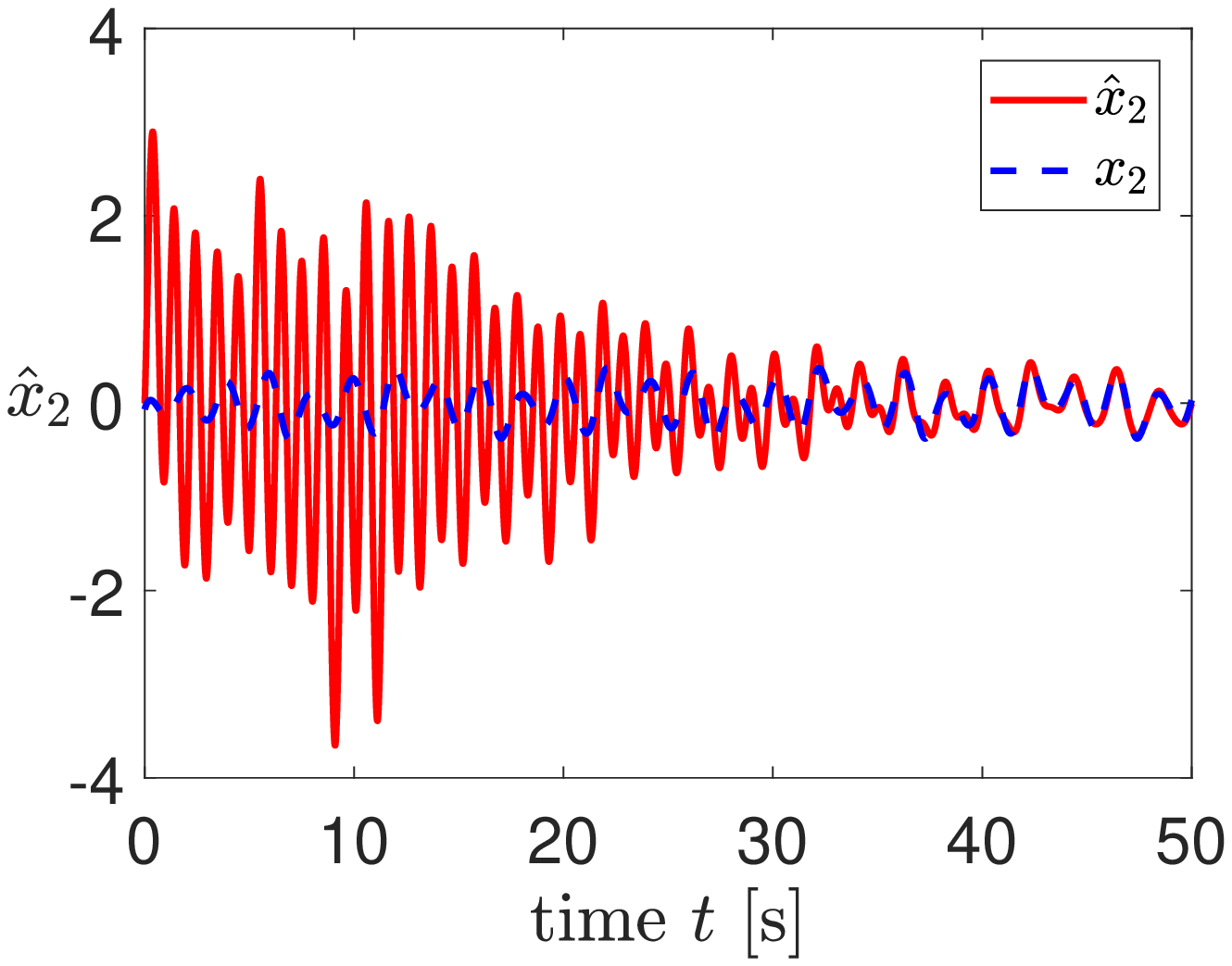}}
    \subcaptionbox{\label{fig: x3_DREM}}{\includegraphics[width=0.23\textwidth]{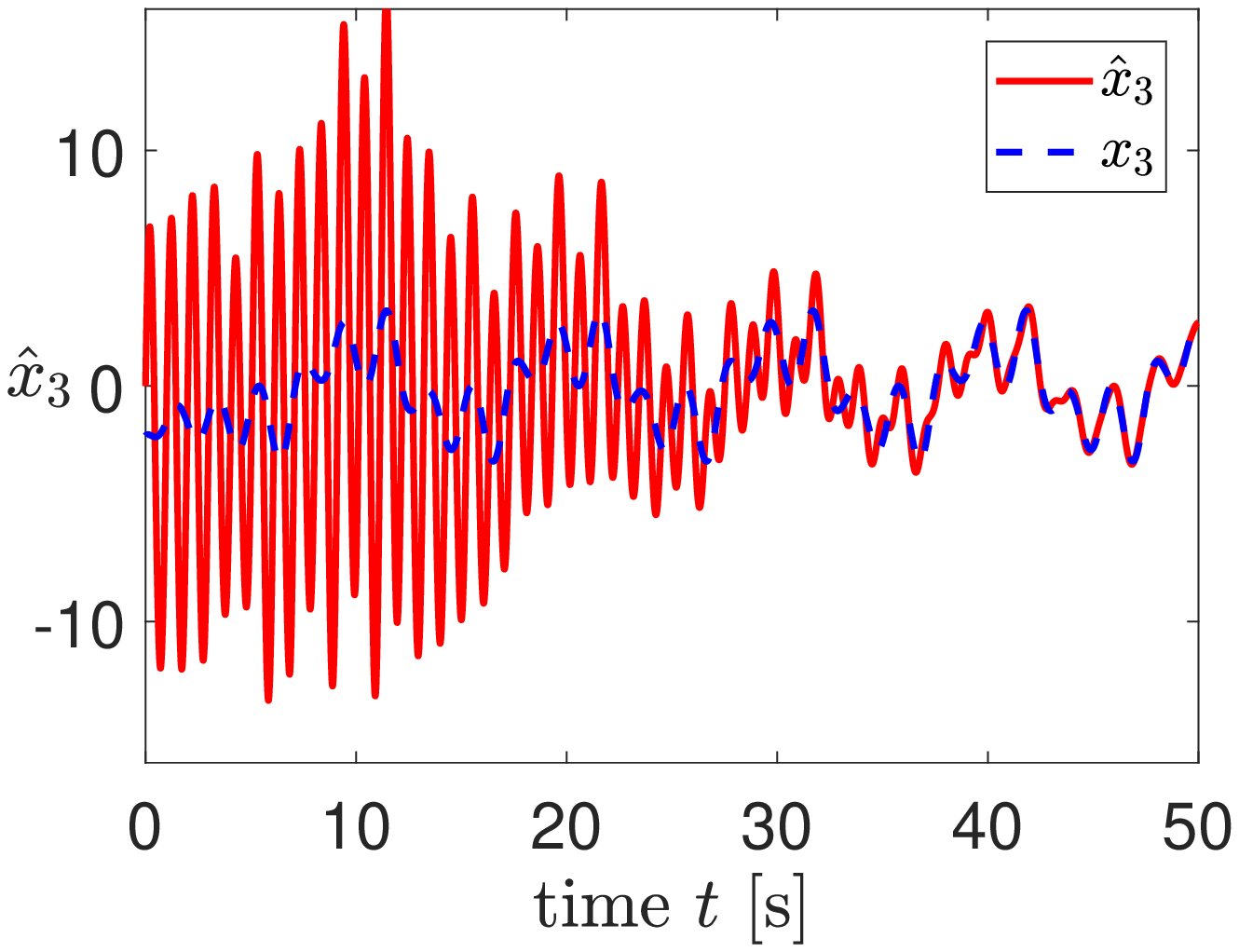}}
    \caption{Parameter and state estimation using the proposed adaptive observer.}
    \label{fig: DREM}
\end{figure}

Finally, we show the robustness of the proposed adaptive observer in the presence of disturbances and noise. The filter parameters $\{K_j\}_{j=1}^5$, the scaling factors $\gamma_i$, as well as the initial value of the observer $\hat{x}$ are specified to be the same as those in the ideal case. The disturbance term in (\ref{eqn: plant_ss_corrupt_1}) is $\omega=[0.05\sin{7 t}, 0.005 \sin{5 t}, 0.1 \sin{13 t}]$. $v$ in (\ref{eqn: plant_ss_corrupt_3}) is generated as random numbers with $|v|\leq v_0=0.1$. In the simulation, the true switching signal $\sigma(t)$ of the plant is used for the switching of the parameter estimator and the state observer.
\begin{figure}[h]
\centering
    \subcaptionbox{\label{fig: theta_DREM_robust}}{\includegraphics[width=0.23\textwidth]{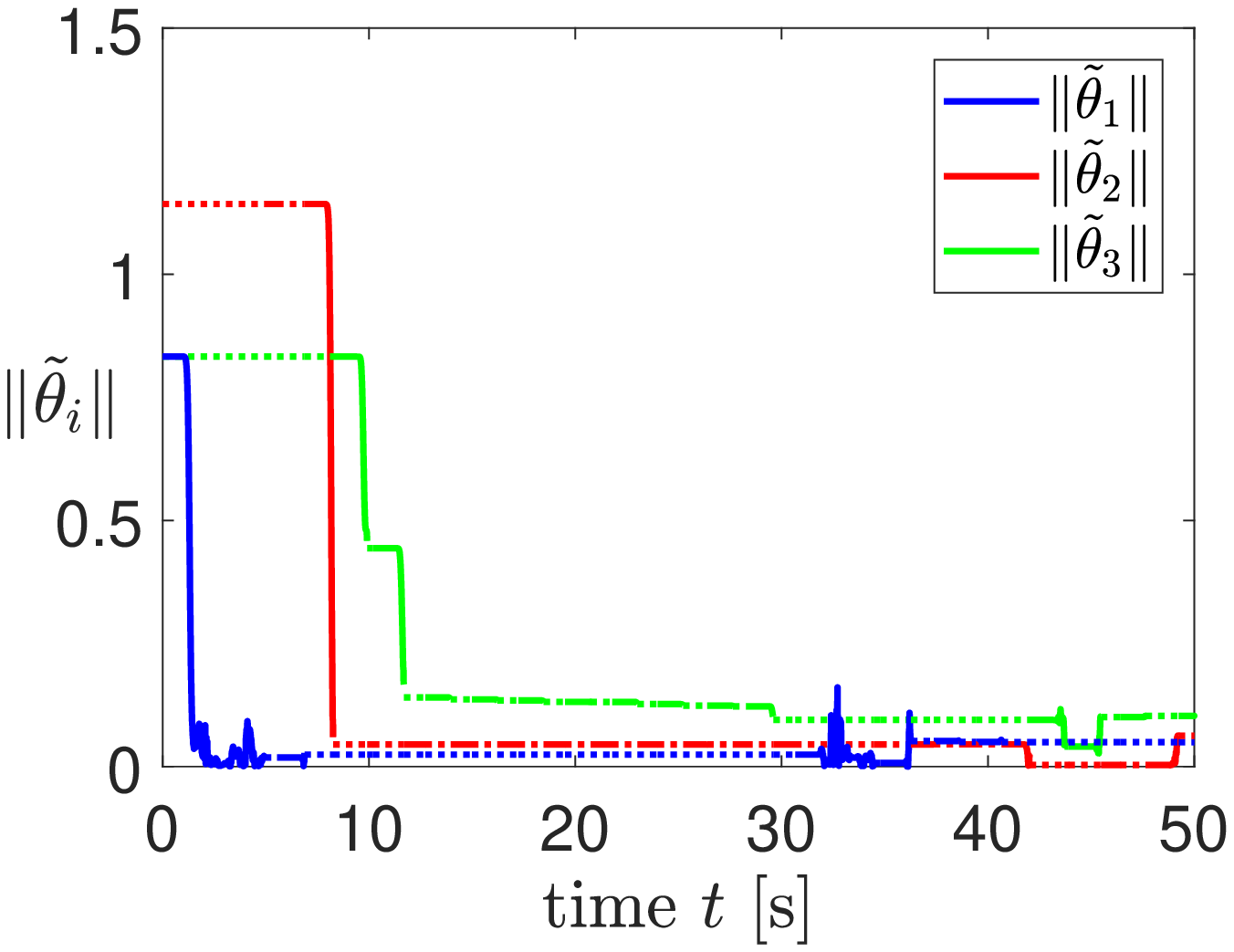}}
    \subcaptionbox{\label{fig: x_DREM_robust}}{\includegraphics[width=0.23\textwidth]{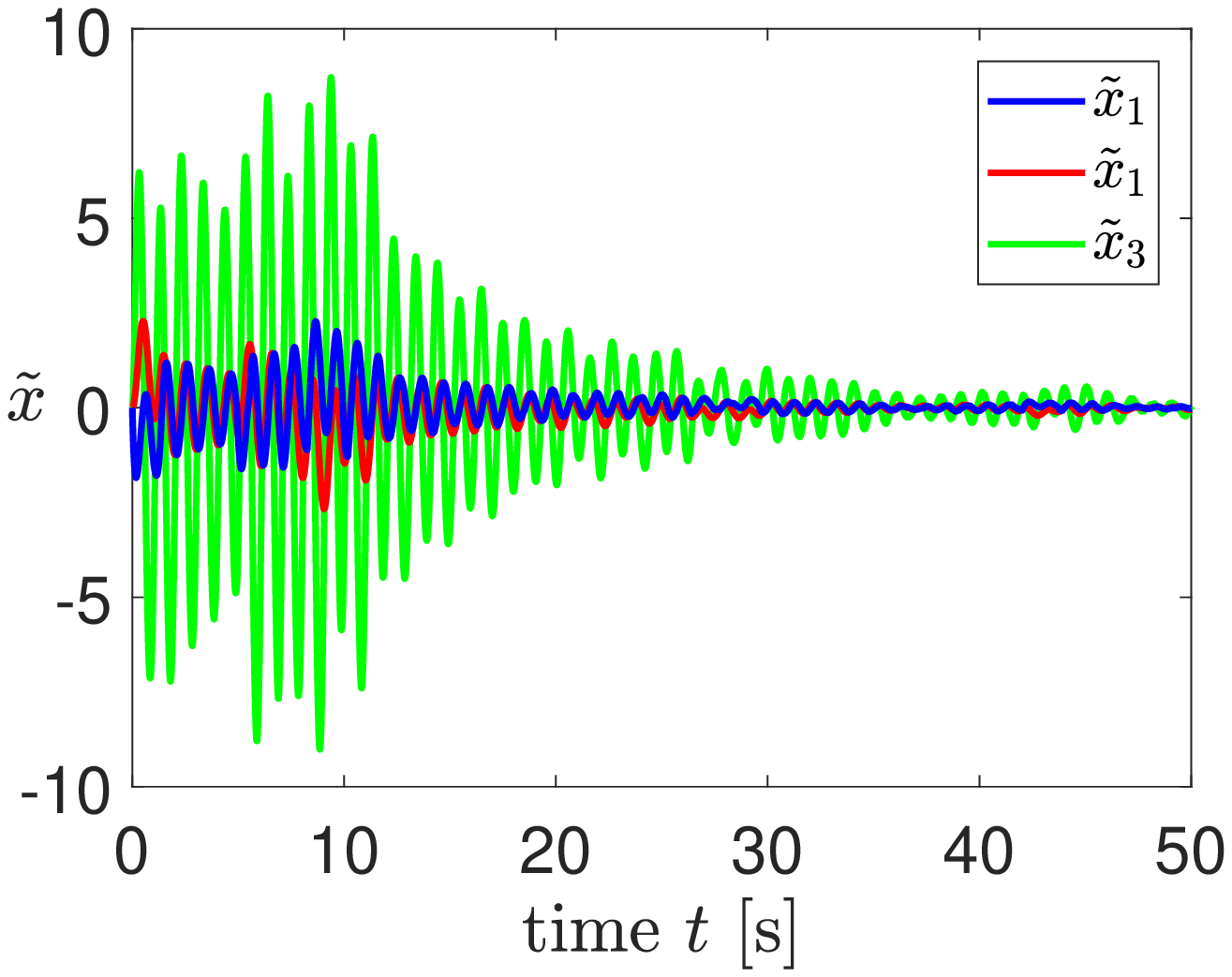}}
    \caption{Parameter and state estimation using the proposed adaptive observer in the presence of disturbances and noise.}
    \label{fig: DREM_robust}
\end{figure}

The parameter estimation errors $\|\tilde{\theta}_i\|$ and the state estimation error $\tilde{x}$ are shown in Fig. \ref{fig: theta_DREM_robust} and Fig. \ref{fig: x_DREM_robust}, respectively. Despite of the presence of disturbances and noise, both the parameter estimation errors and the state estimation error converge to bounded sets, which implies the robustness of the proposed adaptive observer. 

\section{Conclusion}
\label{sec: con}
We have explored the adaptive observer design for a class of systems with switched unknown parameters. Instrumental for this task is the derivation of a new LRE, which takes the intermittently appeared zero-input responses into account. We underscore the novelty that we convert the known information of zero-input responses into a part of an augmented LRE and propose a DREM-based parameter estimator to decouple the parameter adaptation from the SASI of these responses. With the proposed adaptive observer, we managed to eliminate the disturbance effect of the zero-input responses and achieve asymptotic convergence of state and parameter estimation errors. Moreover, we have proved the robustness of the proposed method by showing that the state and parameter estimation errors converge to a bounded set in the presence of disturbances and noise. Future work may include the extension of the proposed adaptive observer to nonlinearly parameterized switched systems.


%

\appendices
\section{Proof of Theorem {\ref{thm: robust}}}
\label{sec: apd_proof}
\begin{proof}
We start the proof by showing that $d_j, j \in \{1,2,\cdots, n+m\}$ in (\ref{eqn: d_j}) is bounded for $t\in[t_0, \infty)$. As $\Phi_j(t,\tau)$ is the state-transition matrix of the Hurwitz matrix $(A- K_j C)$, there exist constants $\beta_j, \beta_j^{\prime} \in \mathbb{R}^+$ such that $\|\Phi_j(t,\tau)\| \leq \beta_j^{\prime} \textrm{e}^{-\beta_j(t -\tau)}$. Therefore, for $t \in [t_k, t_{k+1})$, we have from (\ref{eqn: d_j}) the following inequality
\begin{align*}
    \begin{split}
        &|d_j| \leq \beta_j^{\prime} \|C\| \int_{t_k}^t\textrm{e}^{-\beta_j(t -\tau)}\|\Psi(y,u)-\Psi(\bar{y},u)\| \|\theta_{\sigma(t)}^*\|\textrm{d}\tau\\
    &+\beta_j^{\prime}\|C\|\int_{t_k}^t \textrm{e}^{-\beta_j(t -\tau)} (\|\omega\|+\|K_j\| |v|) \textrm{d}\tau + |v|.
    \end{split}
\end{align*}
Let $L_{\theta}=\max_{i} \|\theta_i\|$ and $\delta t_k=t_{k+1}-t_k, k\in \mathbb{N}$. Due to Assumption \ref{asp: Lipschitz} we obtain for $t \in [t_k, t_{k+1}), k \in \mathbb{N}$
\begin{align*}
    \begin{split}
        |d_j| &\leq \beta_j^{\prime} \|C\| \int_{t_k}^t\textrm{e}^{-\beta_j(t -\tau)} (L_{\Psi} L_{\theta} v_0 +\omega_0+\|K_j\|v_0) \textrm{d}\tau + v_0\\
        &\leq  \frac{\beta_j^{\prime}}{\beta_j} \|C\| (L_{\Psi} L_{\theta} v_0 +\omega_0+\|K_j\|v_0) (1-\textrm{e}^{-\beta_j \delta t_k})+v_0\\
        &< \frac{\beta_j^{\prime}}{\beta_j} \|C\| (L_{\Psi} L_{\theta} v_0 +\omega_0+\|K_j\|v_0)+v_0,
    \end{split}
\end{align*}
which together with $\textrm{adj}(N^T) \in \mathcal{L}_{\infty}$ leads to $D, \bar{D} \in \mathcal{L}_{\infty}$.

From (\ref{eqn: parameter_error_vector_corrupt}), we have
\begin{equation*}
    \tilde{\theta}_{i}(t)=\textrm{e}^{-\gamma_i \int_{t_0}^t \chi_i \Delta^2 (s) \textrm{d}s}\tilde{\theta}_{i}(t_0) + \gamma_i \int_{t_0}^t \textrm{e}^{-\gamma_i \int_{\tau}^t \chi_i \Delta^2(s) \textrm{d}s} \chi_i \Delta D \textrm{d}\tau.
\end{equation*}

Since $\chi_i \Delta$ is PE, there exist constants $\alpha_i^\prime, \alpha_i \in \mathbb{R}^+$ such that
\begin{align}
    \begin{split}
        \|\tilde{\theta}_{i}(t)\| \leq &\alpha_i^\prime \textrm{e}^{-\gamma_i \alpha_i (t-t_0)} \|\tilde{\theta}_{i}(t_0)\|\\
        &+\alpha_i^\prime \gamma_i \int_{t_0}^t \textrm{e}^{-\gamma_i \alpha_i (t-t_0)} \|\chi_i\Delta D\| \textrm{d}\tau,
    \end{split}
\end{align}
which further leads to
\begin{equation}
\label{eqn: theta_bound}
    \lim_{t \to \infty} \sup_{\tau \geq t} \|\tilde{\theta}_{i}(\tau)\| \leq \frac{\alpha_i^\prime}{\alpha_i} \sup_{t} \|\chi_i(t) \Delta(t) D(t)\|.
\end{equation}
From (\ref{eqn: plant_ss_corrupt}) and (\ref{eqn: adaptive observer_corrupt}) we obtain
\begin{align*}
    \begin{split}
        \dot{\tilde{x}}=(A-KC)\tilde{x}+(\Psi(\bar{y},u)\hat{\theta}_{\sigma(t)}-\Psi(y,u)\theta_{\sigma(t)}^*)-w+Kv,
    \end{split}
\end{align*}
which can be further rearranged as 
\begin{align*}
    \begin{split}
        \dot{\tilde{x}}=(A-KC)\tilde{x}&+(\Psi(\bar{y},u)-\Psi(y,u))\hat{\theta}_{\sigma(t)}\\
        &+\Psi(y,u)\tilde{\theta}_{\sigma(t)}-w+Kv.
    \end{split}
\end{align*}
Recalling that $(A-KC)$ is Hurwitz, there exist constants $\beta, \beta^{\prime} \in \mathbb{R}^+$ such that 
\begin{align*}
    \begin{split}
        \|\tilde{x}(t)\|& \leq \beta^{\prime} \textrm{e}^{-\beta(t-t_0)}\|\tilde{x}(t_0)\|+\beta^{\prime}\int_{t_0}^t \textrm{e}^{-\beta (t-\tau)} (L_\Psi v_0 \|\hat{\theta}_{\sigma(t)}\|\\
        &+\sup_t \|\Psi(y,u)\| \|\tilde{\theta}_{\sigma(t)}\| + w_0 + \|K\| v_0) \textrm{d}\tau
    \end{split}
\end{align*}
Since $\|\hat{\theta}_{\sigma(t)}\|\leq \|\tilde{\theta}_{\sigma(t)}\|+\|\theta_{\sigma(t)}^*\|\leq \|\tilde{\theta}_{\sigma(t)}\|+ L_{\theta}$, we obtain
\begin{align*}
    \begin{split}
        \|\tilde{x}(t)\| &\leq \beta^{\prime} \textrm{e}^{-\beta(t-t_0)}\|\tilde{x}(t_0)\|+\beta^{\prime}\int_{t_0}^t \textrm{e}^{-\beta (t-\tau)} (L_\Psi v_0 \|\tilde{\theta}_{\sigma(t)}\|\\
        &+L_\Psi L_{\theta}v_0+\sup_t \|\Psi(y,u)\| \|\tilde{\theta}_{\sigma(t)}\| + w_0 + \|K\| v_0 )\textrm{d}\tau.
    \end{split}
\end{align*}
This further yields
\begin{align*}
    \begin{split}
        \lim_{t\to \infty} \sup_{\tau \geq t}\|\tilde{x}(\tau)\| \leq &\frac{\beta^{\prime}}{\beta} ((L_\Psi v_0 +\sup_t \|\Psi(y,u)\|) \|\tilde{\theta}_{\sigma(t)}\|\\
        &+L_\Psi L_{\theta}v_0 + w_0 + \|K\| v_0 ),
    \end{split}
\end{align*}
which together with (\ref{eqn: theta_bound}) gives (\ref{eqn: robust_bound}) with
\begin{align}
    \begin{split}
        \mu&=\max_i\frac{\alpha_i^\prime}{\alpha_i}(1+\frac{\beta^\prime}{\beta} (L_\Psi v_0 +\sup_t \|\Psi(y,u)\|)),\\
        c&=\frac{\beta^\prime}{\beta}(L_\Psi L_{\theta}v_0 + w_0 + \|K\| v_0 ).
    \end{split}
\end{align}
This completes the proof.
\end{proof}

\ifCLASSOPTIONcaptionsoff
  \newpage
\fi

\bibliographystyle{unsrt}
\bibliography{ref}

\end{document}